\documentclass[runningheads]{llncs}
\usepackage{graphicx}
\usepackage{amsmath, amsfonts}						
\usepackage{enumerate}								
\usepackage{algorithm}								
\usepackage{booktabs}								
\usepackage{algorithmic}							
\usepackage{xfrac}									
\usepackage{subfigure}								
\usepackage{tcolorbox}								
\usepackage{cite}									
\usepackage{enumitem}
\usepackage[normalem]{ulem}							
\usepackage{color}									
\usepackage{multirow}

\newcommand{\pr}{\ensuremath{\mathrm{Pr}}}
\newcommand{\bn}{\ensuremath{\mathcal{B}}}
\newcommand{\bnpr}{\ensuremath{\mathcal{B}'}}

\newcommand{\yes}{\textit{Yes}}
\newcommand{\no}{\textit{No}}
\newcommand{\mprob}[1] {\pr(#1)}
\newcommand{\cprob}[2] {\pr(#1\;|\;#2)}
\newcommand{\apr}{\ensuremath{\mathit{approxsol}_{\bn}}}
\newcommand{\opt}{\ensuremath{\mathit{optsol}_{\bn}}}
\newcommand{\mopt}{\ensuremath{\mathit{optsol}_{\bn}^{1 \dots m}}}
\newcommand{\can}{\ensuremath{\mathit{cansol}_{\bn}}}

\def\pmap{\textsc{Partial MAP}}
\def\spmap{\textsc{Similar Partial MAP}}
\def\dpmap{\textsc{Dissimilar Partial MAP}}

\def\majsat{\textsc{MajSAT}}

\newcommand{\pt}{\ensuremath{\mathsf{P}}} 
\newcommand{\np}{\ensuremath{\mathsf{NP}}}
\newcommand{\pp}{\ensuremath{\mathsf{PP}}}
\newcommand{\nppp}{\ensuremath{\mathsf{NP^{\mathsf{PP}}}}}

\newcommand{\fppppp}{\ensuremath{\mathsf{FP^{PP^{\mathsf{PP}}}}}}

\def\makefig#1#2#3#4{\begin{figure}[h!] 
   \centering
   \includegraphics[width=#3]{#1.#4}
   \caption{#2}
   \label{#1}
   \rule{\columnwidth}{0.3mm}
   \vspace{1mm}
   \end{figure}}

\def\fproblem#1#2#3{
\begin{flushleft} 
  \noindent 
  {\sc #1}\\
  {\bf Instance: }#2.\\
  {\bf Output: }#3.\\
\end{flushleft}}

\begin{document}
\title{Finding dissimilar explanations in\\Bayesian networks: Complexity results}
%
%
\author{Johan Kwisthout\orcidID{0000-0003-4383-7786}}
\authorrunning{J. Kwisthout}
%
\institute{Donders Institute for Brain, Cognition and Behaviour\\
Radboud University Nijmegen,\\
PO Box 9104, 6500HE Nijmegen, The Netherlands\\
\email{j.kwisthout@donders.ru.nl}\\
\url{http://www.socsci.ru.nl/johank}}
\maketitle              
\begin{abstract}
Finding the most probable explanation for observed variables in a Bayesian network is a notoriously intractable problem, particularly if there are hidden variables in the network. In this paper we examine the complexity of a related problem, that is, the problem of finding a set of {\em sufficiently dissimilar}, yet all plausible, explanations. Applications of this problem are, e.g., in search query results (you won't want 10 results that all link to the same website) or in decision support systems. We show that the problem of finding a `good enough' explanation that differs in structure from the best explanation is at least as hard as finding the best explanation itself.

\end{abstract}
\section{Introduction}

A vital computational problem within probabilistic graphical models such as Bayesian networks is the problem of finding the {\em mode} or {\em most probable explanation} of a set of variables given observed values for other variables in the network. When the network includes latent or hidden variables (i.e., variables that have neither been observed nor are of interest for the explanation) this problem is known as \pmap; the explanation sought is the MAP explanation, i.e., the joint value assignment to the explanation variables that has maximum posterior probability. In this paper we are interested in the problem of finding not specifically the MAP explanation, but {\em other} explanations that have desirable properties. There are two distinct reasons why we may be interested in alternative explanations:

\begin{enumerate}
\item As a means of approximating the MAP explanation. \pmap\ is a highly intractable problem \cite{Campos11,Kwisthout11} and we may find an acceptable approximation thereof (to be further explicated later) `good enough';
\item To obtain `alternative good' explanations {\em in addition to} the MAP explanation, to allow us to explore several alternatives (e.g. search results) or to cover for a set of likely explanations (e.g. medical conditions).
\end{enumerate}

`Good enough' and `alternative good' explanations are conceptually different when we look at the {\em structure} of the explanation as compared to the MAP explanation; that is, how similar or dissimilar the joint value assignments (of the MAP explanation and the alternative explanation) are. In the first case we are more than happy to obtain an explanation that is almost identical to the MAP explanation; in fact, in some problem domains this might even be a prerequisite. For example, in computational cognitive modeling an explanation that does not resemble the MAP explanation would not be acceptable as a valid approximation of the MAP explanation, even if it has a comparable probability. In the second case, our goal is to find alternative explanations that are {\em structurally different} yet are still plausible. For example, in response to a search query we don't want to end up with ten results that refer to minor variations of essentially the same web-page, even if they happen to be the most probable given the query.

Note that {\em structure} approximation (where we seek an explanation with Hamming distance at most $d$ of the MAP explanation) is really different from {\em value} approximation (where we seek an explanation with almost-as-high probability, e.g. within ratio $r$ of the MAP explanation) and from {\em rank} approximation (where we seek an explanation that is within the $m$ best explanations) \cite{Kwisthout15}. These notions are really orthogonal, as Figure \ref{map_structs} will reveal. 

\makefig{map_structs}
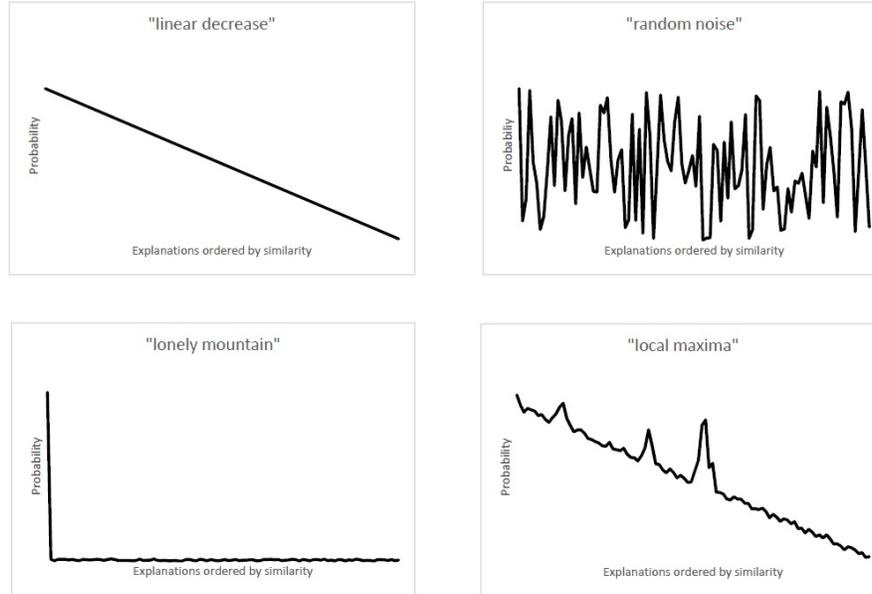
{Graphical depiction of possible relationships between similarity and probability of explanations, with on the X-axis defines an order of the explanations according to their similarity to the best explanation. In ``linear decrease'' the structure and the probability correlate almost completely; in ``random noise'' almost not at all. There can be a sole peak of probability mass (``lonely mountain'') with all other explanations having almost zero probability. Alternatively, there can be several structurally distinct ``local maxima'' of probability mass. Particularly in the latter case it might be interesting to look at dissimilar explanations that have a relatively high probability.}
{12cm}{jpg}

Current algorithms that seek to find alternative explanations by exploring local maxima (e.g. \cite{Chen13,Felzenszwalb10,Chen18}) may fare well if the explanatory landscape is as in the ``local maxima'' panel of Figure \ref{map_structs}. They might not find good explanations, or take a lot of time, if the landscape is in either of the other panels of Figure \ref{map_structs}. However, the computational complexity of this problem has not yet been investigated. In this paper we are specifically interested in explanations that rank well (are within the best $m$ explanations) {\em and} that are either structurally similar or dissimilar; that is, we complement the results of \cite{Kwisthout15}. In this paper we will further explicate both problems and explore the computational complexity of both of them. We will start with offering some necessary preliminaries and sharing our notational conventions. In Section \ref{Main_results} we will formalize both problems (in several variants) and show that both of them are at least as hard as Partial MAP itself. We will also further elaborate on the exact complexity of decision versions of both problems which turns out to be non-trivial. We conclude in Section \ref{Conclusion}. 

\section{Preliminaries}
\label{Preliminaries}

In this section we introduce our notational conventions. Specifically we will cover Bayesian networks, the complexity classes \pp\ and \nppp, one-Turing reductions, and formal definitions of the approximation notions we will use in the paper. The reader is referred to textbooks like \cite{Darwiche09} (specifically complexity in Bayesian networks) and to \cite{Kwisthout15} (for a formal treatment of MAP approximations) for more background.

A Bayesian network $\bn = (\mathbf{G}_{\mathcal{B}}, \pr)$ is a probabilistic graphical model that succinctly represents a joint probability distribution $\mprob{\mathbf{V}} = \prod_{i=1}^n \cprob{V_i}{\pi(V_i)}$ over a set of discrete random variables $\mathbf{V}$. $\bn$ is defined by a directed acyclic graph $\mathbf{G}_{\mathcal{B}} = (\mathbf{V}, \mathbf{A})$, where $\mathbf{V}$ represents the stochastic variables and $\mathbf{A}$ models the conditional (in)dependences between them, and a set of parameter probabilities $\pr$ in the form of conditional probability tables (CPTs). In our notation $\pi(V_i)$ denotes the set of parents of a node $V_i$ in $\mathbf{G}_{\mathcal{B}}$. We use upper case to indicate variables, lower case to indicate a specific value of a variable, and boldface to indicate sets of variables respectively joint value assignments to such a set. 

One of the key computational problems in Bayesian networks is the problem to find the most probable explanation for a set of observations, i.e., a joint value assignment to a designated set of variables (the explanation set) that has maximum posterior probability given the observed variables (the joint value assignment to the evidence set) in the network. If the network includes variables that are neither observed nor to be explained (referred to as intermediate variables) this problem is typically referred to as \pmap. We use the following formal definition:

\fproblem{\pmap} 
{A Bayesian network $\bn = (\mathbf{G}_{\mathcal{B}}, \pr)$, where $\mathbf{V}(\mathbf{G}_{\mathcal{B}})$ is partitioned into a set of evidence nodes $\mathbf{E}$ with a joint value assignment $\mathbf{e}$, a set of intermediate nodes $\mathbf{I}$, and an explanation set $\mathbf{H}$}
{A joint value assignment $\mathbf{h}$ to $\mathbf{H}$ such that for all joint value assignments $\mathbf{h^{\prime}}$ to $\mathbf{H}$, $\cprob{\mathbf{h}}{\mathbf{e}} \geq \cprob{\mathbf{h^{\prime}}}{\mathbf{e}}$}

The following notation is taken from \cite{Kwisthout15}. For an arbitrary \pmap\ instance $\{\bn,\mathbf{H},\mathbf{E},\mathbf{I},\mathbf{e}\}$, let $\can$ refer to the set of candidate solutions, with $\opt\in\can$ denoting the {\em optimal} solution (or, in case of multiple solutions with the same posterior probability, one of the optimal solutions) to the \pmap\ instance; we will informally refer to this solution as the MAP explanation. When $\can$ is ordered according to the probability of the candidate solutions (breaking ties between candidate solutions with the same probability arbitrarily), then $\mopt$ refers to the set of the first $m$ elements in $\can$, viz. the $m$ most probable solutions to the \pmap\ instance.

We assume that the reader is familiar with standard notions in computational complexity theory, notably the classes \pt\ and \np, \np-hardness, and polynomial time (many-one) reductions. The class \pp\ is the class of decision problems that can be decided by a probabilistic Turing machine in polynomial time; that is, where \yes-instances are accepted with probability strictly larger than $\sfrac{1}{2}$ and \no-instances are accepted with probability no more than $\sfrac{1}{2}$. A problem in \pp\ might be accepted with probability $\sfrac{1}{2} + \epsilon$ where $\epsilon$ may depend exponentially on the input size $n$. Hence, it may take exponential time to increase the probability of acceptance (by repetition of the computation and taking a majority decision) close to $1$. This is consistent with the sampling variant of the Chernoff bound: The number of samples $M$ needed to increase the probability of acceptance of \yes-instances to $1-\delta$ is at least $\frac{\ln(\sfrac{1}{\sqrt{\delta}})}{\epsilon^2}$; when $\epsilon = \sfrac{1}{2^n}$ then $M$ is exponential in the input size. \pp\ hence is a powerful class; we know for example that $\np\subseteq\pp$ and the inclusion is assumed to be strict. The canonical \pp-complete decision problem is \majsat: given a Boolean formula $\phi$, does the majority of truth assignments to its variables satisfy $\phi$?

In computational complexity theory, so-called {\em oracles} are theoretical constructs that increase the power of a specific Turing machine. An oracle (e.g., an oracle for \pp-complete problems) can be seen as a `magic sub-routine' that answers class membership queries (e.g, in \pp) in a single time step. In this paper we are specifically interested in classes defined by non-deterministic Turing machines with access to a \pp-oracle. Such a machine is very powerful, and likewise problems that are complete for the corresponding complexity class \nppp\ are highly intractable. 

A decision variant of \pmap\ is known to be \nppp-complete, even for binary variables, indegree at most $2$, and under the assumption that there exists at least one joint value assignment $\mathbf{h}$ such that $\mprob{\mathbf{h,e}} > 0$ \cite{Park04,Kwisthout11a}. In the intractability proofs in Section \ref{Main_results} we will assume, without loss of generality, that these constraints hold for all Bayesian networks $\bn$ under consideration. As we use reductions from function problems, not decision problems, our reductions are formally polynomial-time one-Turing reductions. A function $f$ one-Turing reduces to $g$ (notation $f \leq^{\mathsf{FP}}_{1-T} g$) if there are functions $t_1$ and $t_2$ such that for all $x, f(x) = t_1(x, g(x, t_2(x)))$ \cite[p.5]{Toda94}.

We finish this section be repeating the following formal definition of rank-approximation of \pmap\ from \cite{Kwisthout15}; we will build on this definition in the next section.

\begin{definition}[rank-approximation of \pmap]
Let $\mopt\subseteq\can$ be the set of the $m$ most probable solutions to a \pmap\ problem and let $\opt$ be the optimal solution. An explanation $\apr\in\can$ is defined to $m$-rank-approximate $\opt$ if $\apr\in\mopt$.
\end{definition}


\section{Main results}
\label{Main_results}

Let $d_H$ be the Hamming distance between two joint value assignments. We define the following two problem variants to \pmap, where $m$ and $d$ are arbitrary constants:

\fproblem{$(m,d)$-\spmap} 
{As in \pmap}
{An explanation $\apr\in\can$ that $m$-rank-approximates $\opt$ and where $d_H(\apr, \opt) \leq d$, or special symbol $\emptyset$ if no such explanation exists}

\fproblem{$(m,d)$-\dpmap} 
{As in \pmap}
{An explanation $\apr\in\can$ that $m$-rank-approximates $\opt$ and where $d_H(\apr, \opt) \geq d$, or special symbol $\emptyset$ if no such explanation exists}

We will prove that both problems are \nppp-hard by reduction from \pmap, even for binary variables with indegree at most $2$. We will start with the construction for $(m,d)$-\spmap\ without the latter constraints and prove \nppp-hardness (Figure \ref{construct}, panel a), and then adapt it to contain only binary variables and indegree at most $2$ (panel b). Then we show how this construction can be extended to prove \nppp-hardness of $(m,d)$-\dpmap\ (panel c). 

\makefig{construct}
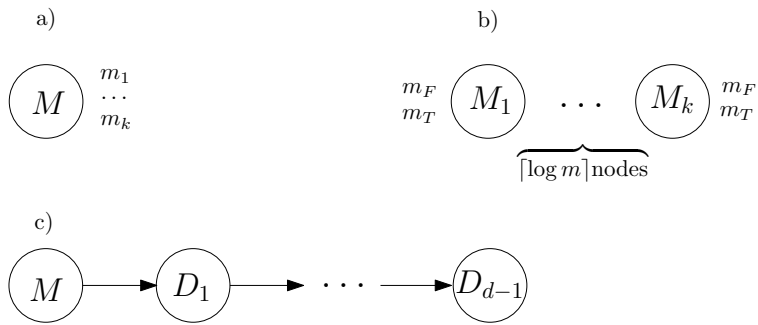
{Nodes added to \pmap\ network \bn. a) basic construct; b) construct with binary nodes only; c) assuming at least $d$ different nodes.}
{10cm}{pdf}



\begin{theorem}
$(m,d)$-\spmap\ is \nppp-hard.
\label{spmap_thm}
\end{theorem}

\begin{proof}
We reduce from the \nppp-hard \pmap\ problem. Let $\{\bn,\mathbf{H},\mathbf{E},\mathbf{I},\mathbf{e}\}$ be an instance to \pmap. From \bn, we create an instance \bnpr\ to $(m,d)$-\spmap\ as follows. To \bn\ we add a singleton, unconnected node $M$ with $m$ uniformly distributed values $m_1 \ldots m_m$. Let $\mathbf{H'} = \mathbf{H} \cup \{M\}$ and observe that the MAP explanation $\mathbf{h}$ to $\mathbf{H}$ in the \pmap\ instance now translates into a set of $m$ best explanations $\mopt = \mathbf{h} \cup \{m_i\} (i = 1 \ldots m)$ with equal probability and by construction all these explanations differ only in the value assignment to $M$. Let $\opt$ be any arbitrary explanation in $\mopt$ and let $apr$ be a distinct explanation in $\mopt$. We have that $\apr\in\mopt$ $m$-rank-approximates $\opt$ and that for any $d \geq 1$, $d_H(\apr, \opt) = 1 \leq d$;  obviously the one-Turing reduction takes polynomial time and hence $(m,d)$-\spmap\ is \nppp-hard.
\end{proof}

Note that we can replace $M$ in this construction by $k = \lceil \log m \rceil$ uniformly distributed binary variables $M_i$ such that $\mathbf{H'} = \mathbf{H} \cup \{M_1 \ldots M_k \}$. A caveat here is that $\mathbf{h}$ may translate to more than $m$ explanations with equal probability and if $\mopt$ and $\opt$ are picked randomly from this set, we may end up\footnote{An example of such a situation would be when $m = 5$ and the solutions with binary encodings $000$, $011$, $101$, $110$, and $111$ would be in $\mopt$, with $\opt = 000$.} with a set that does not contain $\apr$ such that $d_H(\apr, \opt) = 1$. We therefore impose the constraint that $\mopt$ does not arbitrarily break ties, but that it contains the first $m$ explanations according to their lexicographical order, and because of this we can be sure that $\mopt$ contains at least one explanation $\apr$ with $d_H(\apr, \opt) = 1$.

\begin{corollary}
$(m,d)$-\spmap\ is \nppp-hard, even if all nodes are binary and have indegree at most 2.
\end{corollary}

We now extend the construction in the proof of theorem \ref{spmap_thm} to prove \nppp-hardness of $(m,d)$-\dpmap. To the above constructed network \bnpr\ we add $d-1$ nodes $D_1 \ldots D_{d-1}$ with $m$ uniformly distributed values. $D_1$ has $M$ as sole parent, whereas $D_i, i \geq 2$ has $D_{i-1}$ as sole parent. The probability of $D_i$ is defined to be deterministically dependent on its parent, i.e., $\cprob{D_1 = d_{1,j}}{M = m_j} = 1$ and $\cprob{D_i = d_{i,j}}{D_{i-1} = d_{i-1,j}} = 1$. We set $\mathbf{H'} = \mathbf{H} \cup \{M\} \cup \{D_1 \ldots D_{d-1}\}$. By construction, every explanation in $\mopt$ will differ in $d$ variables and thus any $\apr\in\mopt$ $m$-rank-approximates $\opt$ and $d_H(\apr, \opt) = d$. The construction with $k = \lceil \log m \rceil$ binary variables $M_i$ is similar, but now we add $(d-1)k$ nodes $D_{1,1} \ldots D_{d-1,k}$, which implies that every $\apr\in\mopt$ differs in at least $d$ nodes.

\begin{corollary}
$(m,d)$-\dpmap\ is \nppp-hard, even if all nodes are binary and have indegree at most 2.
\end{corollary}



Note that $m$ and $d$ are constants and not part of the input. If we would make them part of the input, then unary notation would be necessary as the number of nodes added would otherwise be exponential in the binary representation of $m$ and $d$. However, since $m \leq 2^n$ (where $n$ denotes the number of variables in \bn) the reduction would then not be polynomial in the \pmap\ instance any more. Hence, we require $m$ (and $d$) to be a constant.

\subsection{On membership in \nppp}
\label{On_membership_in_nppp}

As explicated before, \pmap\ has an \nppp-complete decision variant \cite{Park04}. Is it likely that appropriate decision variants of \spmap\ and \dpmap\ are also in \nppp? We can check in polynomial time, using the MAP explanation \opt\ and the approximation $\apr\in\can$, that the Hamming distance $d_H(\apr, \opt)$ is as promised. However, a verification algorithm for $(m,d)$-\spmap\ and $(m,d)$-\dpmap\ should also verify that \apr\ actually is within the $m$ best explanations, and finding the $m$-th best explanation is known to be \fppppp-complete \cite{Kwisthout11a}. Hence, it is unlikely that $(m,d)$-\spmap\ and $(m,d)$-\dpmap\ have an \nppp-complete decision variant.

\section{Conclusion}
\label{Conclusion}

In this short paper we elaborated on the computational complexity of finding MAP explanations that are almost-as-good as the most probable one (where `almost-as-good' was defined by rank) {\em and} that are sufficiently similar or dissimilar to the most probable explanation. We found that these problems are \nppp-hard, that is, not easier than the \pmap\ problem itself. An attempt to extend this proof construct to include value-approximation did not succeed, as the proof constructs really require that all explanations in $\apr\in\mopt$ have the same probability. In an extreme case it might be that we have just two explanations for a given \pmap\ problem, where $\cprob{h}{\mathbf{e}} = \sfrac{1}{2} + \epsilon$ and $\cprob{\neg h}{\mathbf{e}} = \sfrac{1}{2} - \epsilon$, and any non-uniform probability assignment to the additional variables in the explanation set will destroy the proof.

In future work we'd like to extend our results to value-approximation as well as rank-approximation. It might be relevant to investigate the complexity of these problems in constrained structures, such as polytrees, where the \pmap\ problem nonetheless remains \np-hard.

\bibliographystyle{splncs04}
\bibliography{Kwisthout}

\begin{thebibliography}{10}
\providecommand{\url}[1]{\texttt{#1}}
\providecommand{\urlprefix}{URL }
\providecommand{\doi}[1]{https://doi.org/#1}

\bibitem{Campos11}
de~Campos, C.P.: New complexity results for {MAP} in {B}ayesian networks. In:
  Walsh, T. (ed.) Proceedings of IJCAI, vol. 11 (2011)

\bibitem{Chen13}
Chen, C., Kolmogorov, V., Zhu, Y., Metaxas, D., Lampert, C.: Computing the {M}
  most probable modes of a graphical model. In: Proceedings of the 16th
  International Conference on Articial Intelligence and Statistics (AISTATS)
  (2013)

\bibitem{Chen18}
Chen, C., Yuan, C., Ye, Z., Chen, C.: Solving {M}-modes in loopy graphs using
  tree decompositions. In: Kratochv\'{i}l, V., Studen\'{y}, M. (eds.)
  Proceedings of the Ninth International Conference on Probabilistic Graphical
  Models. Proceedings of Machine Learning Research, vol.~72, pp. 145--156
  (2018)

\bibitem{Darwiche09}
Darwiche, A.: Modeling and Reasoning with {B}ayesian Networks. Cambridge
  University Press (2009)

\bibitem{Felzenszwalb10}
Felzenszwalb, P., Girshick, R., McAllester, D., Ramanan, D.: Object detection
  with discriminatively trained part-based models. IEEE Transactions on
  Software Engineering  \textbf{32}(9),  1627--1645 (2010)

\bibitem{Kwisthout11}
Kwisthout, J.: Most probable explanations in {B}ayesian networks: Complexity
  and tractability. International Journal of Approximate Reasoning
  \textbf{52}(9),  1452--1469 (2011)

\bibitem{Kwisthout15}
Kwisthout, J.: Tree-width and the computational complexity of {MAP}
  approximations in {B}ayesian networks. Journal of Artificial Intelligence
  Research  \textbf{53},  699--720 (2015)

\bibitem{Kwisthout11a}
Kwisthout, J., Bodlaender, H.L., van~der Gaag, L.C.: The complexity of finding
  kth most probable explanations in probabilistic networks. In: Proceedings of
  the 37th International Conference on Current Trends in Theory and Practice of
  Computer Science (SOFSEM 2011). vol. LNCS 6543, pp. 356--367. Springer (2011)

\bibitem{Park04}
Park, J., Darwiche, A.: Complexity results and approximation settings for {MAP}
  explanations. Journal of Artificial Intelligence Research  \textbf{21},
  101--133 (2004)

\bibitem{Toda94}
Toda, S.: Simple characterizations of {P}(\#{P}) and complete problems. Journal
  of Computer and System Sciences  \textbf{49},  1--17 (1994)

\end{thebibliography}

\end{document}